\documentclass[10pt]{article}

\usepackage{authblk}
\usepackage{caption}
\usepackage{amstext,amsmath,amssymb,amsfonts,bbm}
\usepackage[latin1]{inputenc}
\usepackage{epsfig}
\usepackage{dsfont}
\usepackage{hyperref}
\usepackage{amsthm}
\usepackage{subfigure}
\usepackage{color}
\usepackage{multirow}
\usepackage{psfrag}
\usepackage{graphicx}
\usepackage{extarrows}
\usepackage{braket}
\usepackage{empheq}

\newcommand*\widefbox[1]{\fbox{\hspace{2em}#1\hspace{2em}}}

\usepackage[lmargin=80pt,rmargin=80pt,tmargin=80pt,bmargin=80pt]{geometry}

\theoremstyle{plain}

\newtheorem{theorem}{Theorem}

\newtheorem{proposition}{Proposition}
\theoremstyle{definition}

\newcommand{\cG}{{\cal G}}

\newcommand{\sgn}{{\rm sgn}}

\newcommand{\be}{\begin{equation}}
\newcommand{\ee}{\end{equation}}

\allowdisplaybreaks[4]

\begin{document}

\title{\bf The $\imath \epsilon$ prescription in the SYK model}

\author[1]{Razvan Gurau}

\affil[1]{\normalsize\it Centre de Physique Th\'eorique, \'Ecole Polytechnique, CNRS, F-91128 Palaiseau, France
and Perimeter Institute for Theoretical Physics, 31 Caroline St. N, N2L 2Y5, Waterloo, ON, Canada. \authorcr
email: rgurau@cpht.polytechnique.fr \authorcr \hfill}

\date{}

\maketitle

\hrule\bigskip

\begin{abstract}
We introduce an $\imath \epsilon$ prescription for the SYK model both at finite and at zero temperature.
This prescription regularizes all the naive ultraviolet divergences of the model. As expected the prescription breaks the 
conformal invariance, but the latter is restored in the $\epsilon \to 0$ limit.
We prove rigorously that the Schwinger Dyson equation of 
the resummed two point function at large $N$ and low momentum is recovered in this limit. 
Based on this $\imath \epsilon$ prescription  we introduce an effective field theory Lagrangian for the infrared SYK model.  
\end{abstract}

\bigskip

\hrule\bigskip

\tableofcontents

\bigskip

\section{Introduction and discussion}

The Sachdev--Ye--Kitaev (SYK)  model \cite{Sachdev:1992fk,Kitaev,Maldacena:2016hyu,Polchinski:2016xgd,Fu:2016vas,Gross:2016kjj,Das:2017pif} has been extensively 
studied recently in the context of the AdS/CFT duality. In its most common form, the SYK model is the one dimensional field theory for 
a vector Majorana fermion $\chi_a$ with $N$ components with action:
\begin{equation}\label{eq:action}
 \frac{1}{2}\int_{-\beta/2}^{\beta/2} d \tau  \sum_{ a} \chi_{a}  (\tau) \partial_{\tau} \chi_{a} (\tau)  + J
 \sum_{a^1,\dots a^q} T_{a^1\dots a^q} \int_{-\beta/2}^{\beta/2} d\tau   \; \chi_{ a^1} (\tau) \dots   \chi_{a^q}(\tau) \;, 
\end{equation}
where $T$ are time independent quenched random couplings with Gaussian distribution:
\[
 d\nu(T) = \bigg( \prod_{a^1,\dots a^q}  \sqrt{ \frac{N^{D-1} }{2\pi} }d T_{a^1,\dots a^q}   \bigg) \; e^{-\frac{1}{2} N^{q-1} \sum_{a^1,\dots a^q} T_{a^1\dots a^q} T_{a^1\dots a^q}  } \;.
\]

This model has a large $N$ limit dominated by melonic graphs \cite{Maldacena:2016hyu,Polchinski:2016xgd,Jevicki:2016bwu}. 
The melonic large $N$ limit is universal in random tensors \cite{RTM}, and the quenching can be eliminated
if one considers a tensor version of the SYK model \cite{Witten:2016iux,Gurau:2016lzk,Klebanov:2016xxf,Peng:2016mxj,Krishnan:2016bvg,Peng:2017kro}
(see also \cite{Bonzom:2017pqs} for a detailed discussion of the leading and next to leading orders in $1/N$ in various models).

Leaving aside the details of the model, the melonic large $N$ limit leads to an ``almost conformal'' one dimensional filed theory.
This theory (the CFT side of the AdS/CFT) has been studied \cite{Maldacena:2016hyu,Polchinski:2016xgd,Gross:2017hcz}
with various degrees of rigor.

This paper aims to give a rigorous meaning to some of the results obtained so far in this research program.

\paragraph{The trouble with the two point function.} Let us briefly review some standard results on the SYK model.
Having a $q$ fermion interaction and a free propagator: 
\[ C(\tau,\tau') = \frac{1}{2}\sgn(\tau-\tau') \; , \]
with antiperiodic boundary conditions at finite temperature, the model defined by Eq.~\eqref{eq:action} is power counting super renormalizable: there are no ultraviolet (UV) divergences, 
and infrared (IR) divergences might exist only at zero temperature. One can then resum the two point function at leading order in $N$.
This resummed two point function, ${\cG}_{\beta}(\tau,\tau')$, is recovered from the Schwinger Dyson equation (SDE):
\[
 1 = {\cG}_{\beta} C^{-1} - {\cG}_{\beta} \Sigma_{\beta} \;,
\]
taking into account that in the melonic large $N$ limit the self energy factors in terms of two point functions $\Sigma_{\beta}(\tau,\tau') = J^2  [{\cG}_{\beta}(\tau,\tau') ]^{q-1}$:
\begin{align*}
 \delta(\tau_1 -\tau_2)  & = \partial_{\tau_1} {\cG}_{\beta}(\tau_1 - \tau_2) -J^2 \int_{-\beta/2}^{\beta/2} du \; \; {\cG}_{\beta}(\tau_1-u) [ {\cG}_{\beta}(u-\tau_2)]^{q-1} \;,
\end{align*}
where we used the fact that ${\cG}_{\beta}$ is antisymmetric and translation invariant. 

While the Schwinger Dyson equation can not be solved analytically at arbitrary momentum (except for the degenerate $q=2$ case \cite{Maldacena:2016hyu}),
a solution can be found in the conformal (low momentum, infrared)  limit. Indeed, in this limit the first term (free term) can be neglected and 
the SDE becomes:
\begin{align}\label{eq:sde}
 \delta (\tau)  = J^2 \int_{- \beta/2}^{\beta/2} du\;  G_{\beta}(  u -\tau  )  \big[ G_{\beta}(u) \big]^{q-1} \; ,
\end{align}
where $G_{\beta} $ denotes the infrared two point function.
Let us, for now, consider the zero temperature case, $\beta \to \infty$ (we will reinstate the finite temperature later on).
In order to solve for the infrared resummed two point function one proposes the ansatz:
\[
 G_{\infty}(\tau) = b \frac{ \sgn(\tau)  }{| \tau |^{2\Delta} } \;,
\]
with $\Delta>0$. Substituting this in Eq.~\eqref{eq:sde} one gets \cite{Maldacena:2016hyu,Polchinski:2016xgd,Klebanov:2016xxf} the equation:
\begin{align*}
&  \delta(\tau)= J^2 b^q \int_{-\infty}^{\infty} du \; \frac{ {\rm sgn}( u -\tau ) }{| u -\tau |^{ 2 \Delta}  } \frac{ {\rm sgn}( u  ) }{|u |^{ 2\Delta(q-1) } }  
  = J^2 b^q \frac{ 1  }{|\tau|^{2 \Delta q -1 }} \times \\
& \quad \times \bigg[   \beta(1- 2\Delta, 2\Delta q -1 ) + \beta(1- 2\Delta(q-1), 2\Delta q-1) - \beta(1- 2\Delta,1- 2\Delta(q-1) ) \bigg] \;, \nonumber
\end{align*}
with $\beta(a,b)$ the Euler beta function. This equation is \emph{formally} solved by $\Delta = \frac{1}{q}$ and $b$ respecting:
\[
 1   = J^2 b^q   \frac{\pi}{ \frac{1}{2} -\frac{1}{q}} \; \frac{\cos   \frac{\pi}{q}  }{ \sin \frac{\pi}{q}   } \;,
\]
however it is quite obvious that:
\begin{itemize}
 \item the left hand side of the equation, $ \frac{ 1  }{|\tau|^{2 \Delta q -1 }} $, is  not a $\delta(\tau)$ function, even for $\Delta=\frac{1}{q}$.
 \item the integral does not converge absolutely in the $u\sim 0$ region for $\Delta = \frac{1}{q}$, as $2\Delta(q-1) = 2 - \frac{2}{q}>1$. 
 This translates on the right hand side in the fact that the Euler beta functions are evaluated at negative arguments. While, of course, the beta function can be defined by analytic continuation 
 at negative arguments, its naive integral representation diverges for such values.
\end{itemize}

The situation only gets worse when one tries to compute the leading order four point function, the spectrum of the four point kernel (which generates the ladder 
diagrams) \cite{Maldacena:2016hyu,Polchinski:2016xgd} or the leading order six point functions \cite{Gross:2017hcz}: all the integrals one encounters exhibit UV divergences.
This should come as no surprise: in the conformal limit the theory is power counting marginal (as one would expect from a conformal field theory).

Of course these divergences have already been noted and discussed in the literature \cite{Maldacena:2016hyu,Polchinski:2016xgd}. Physically, they are regulated 
by the fact that at large momentum one can not use the conformal ansatz  $G_{\beta} $ and one must go back to the full two point function ${\cG}_{\beta}$. 
Using the full two point function regulates all the divergences of the model: after all, we already know that the model is UV finite.
However, as the SDE can not be solved analytically at arbitrary momentum, one does not have an explicit formula for ${\cG}_{\beta}$. In the absence
of such a formula, the procedure applied so far \cite{Maldacena:2016hyu,Polchinski:2016xgd} consists in 
the following:
\begin{description}
 \item[\bf In most cases.] In most cases one can try to make sense of these integrals by analytic continuation. 
   One can hope that, due to the antisymmetry of the two point function, all the UV divergences are regulated if one defines the integrals
   by, for instance, a Cauchy principal value. In practice one computes the integrals for values of the parameters (like for instance $\Delta$)
   for which they converge and then substitutes the relevant values (like $\Delta =\frac{1}{q}$) only at the end. Typically this leads to some Euler $\Gamma(a)$ functions evaluated at arguments 
   $a$ with negative real part which are well defined by analytic continuation.
   However this approach has several drawbacks:
  \begin{itemize}
    \item sometimes one needs to formally evaluate integrals which are divergent for \emph{any} values of the parameters \cite{Polchinski:2016xgd}, therefore not even the starting point of 
    the analytic continuation is well defined.
    \item the classical integral  representation for the $\Gamma(a)$ function at $\Re(a)<0$ requires \cite{Bergere:1977ft,Bergere:1980sm} counterterm subtractions:
       \begin{align*}
         -p-1 <\Re(a)<- p \; ,\qquad \Gamma(a) = \int_{0}^{\infty} dt\; t^{a-1} \left( e^{-t} -\sum_{q=0}^{p} \frac{(-t)^p}{p!} \right) \;.
       \end{align*}
       It is not clear where the counterterms might come from.
    \item the fact that the two point function is antisymmetric does not eliminate the UV divergences. Indeed, if two vertices of a graph
     are connected by and even number of edges larger or equal to $q/2$, the corresponding integral is divergent and symmetric hence the graph is
     UV divergent\footnote{One can still attempt to deal with this by resuming families of graphs. This is a formal manipulation, as each individual graph 
     in the family is divergent. Moreover, except is very simple cases, one can not identify appropriate families of graphs to (formally) cancel all the divergences.}.
   \item finally, and most importantly, in the absence of an explicit regularization procedure, there is a priori no reason to consider the Cauchy principal value
   in the first place. In fact it turns out that the $\imath \epsilon$ regularization we introduce in this paper justifies the use of the Cauchy principal 
   value in some of the cases encountered in \cite{Maldacena:2016hyu,Polchinski:2016xgd}.
  \end{itemize}
 
 \item[\bf In some cases.] In some cases the above procedure fails. This is notably the case (using the notation of \cite{Maldacena:2016hyu})
 of  the $h=2$ mode of the four point kernel which leads to a breaking of conformal invariance in the resummed leading order four point function.
 In this case  the UV divergences are crucial and one needs to deal with them carefully. The procedure applied so far \cite{Maldacena:2016hyu} 
 (also discussed to a lesser extent in \cite{Polchinski:2016xgd})
 is to account for the effect of the free term in the SDE using first order perturbation theory in quantum  mechanics. This has several drawbacks:
 \begin{itemize}
  \item  while first order perturbation theory in quantum mechanics eliminates the divergence, it is difficult to see in what sense 
  such a regularization can be rendered rigorous (the perturbation theory in quantum mechanics usually diverges).
  \item  it is not a priori obvious that this procedure will regulate all the divergences.
  \item  perturbation theory in quantum mechanics is model dependent. In order to study the departure from conformality in the SYK model in a systematic manner, a 
  more appropriate starting point would be a universal regularization procedure.  
 \end{itemize}
\end{description}

In this paper we propose an $\imath \epsilon$ prescription for the SYK model which regulates all the UV divergences. The limit $\epsilon\to 0$ can be taken rigorously. 
Our prescription is a particular kind of cutoff in the frequency space and comes to replacing the low momentum resummed two point function  $G_{\beta}$ by a regulated version $G^{\epsilon}_{\beta}$. 
Like the full two point function $\cG_{\beta}$ of the SYK model, the regulated two point function $G^{\epsilon}_{\beta}$ breaks the conformal invariance. Contrary to 
$\cG_{\beta}$ however, $G^{\epsilon}_{\beta}$ does this in an universal manner.

The interpretation of this prescription is best understood if one takes a quantum field theoretical point of view on the SYK model.
The momentum scale at which one feels the breaking of conformal invariance due to the first term of the SDE, where one should start using $\cG_{\beta}$ instead of $G_{\beta}$,
plays the role of an ultimate ``physical cutoff scale''.
In the case of quantum electrodynamics (QED) for instance this should be taken as the 
scale at which quantum chromodynamics (QCD) effects come into play; for the standard model as a whole this could be a grand unification scale, or the Plank scale. 
Its precise value, and the precise way in which it alters the UV behavior of the model should play no role in understanding the departure from conformality in the SYK model
(to pursue our comparison, understanding that QED flows to the Gaussian fixed point in the infrared and computing the $\beta$ function close to the Gaussian fixed point 
does not depend on the number of quark generations).
In order to understand the infrared behavior of the model one needs to introduce a new scale (call it the ``mathematical cutoff scale'') and a regularization procedure (for instance a 
multiplicative momentum cutoff or a Schwinger parametric cutoff).
This is an arbitrary UV scale, which can be considered lower that the physical cutoff scale (in QED this would be 
a cutoff scale in the neighborhood of the Gaussian fixed point). Introducing this new scale and 
a regularization procedure at this scale allows one to ignore the true UV completion of the theory (in QED, once one introduces a UV cutoff scale, one
ignores the rest of the standard model), and study its infrared behavior in a self contained manner.

The $\imath \epsilon$ prescription we present here yields the ``mathematical cutoff'' of the SYK model. It allows one to study the departures from conformality 
without needing to resort to the precise UV completion $\cG_{\beta}$ of the model. The $\epsilon$ scale is a mathematical artifact which 
identifies the overall power counting of an effect, but there is no meaning attached to the specific value of $\epsilon$.

An upshot of our $\imath \epsilon$ prescription is that we are able, at finite $\beta$, to write down an explicit effective field theory Lagrangian  whose large $N$
resummed two point function is $G^{\epsilon}_{\beta}$. The similarity between the Lagrangian we propose here and the 
``conformal SYK'' Lagrangian recently discussed in \cite{Gross:2017vhb} is only superficial: the two models differ drastically in the infrared. 
To be precise, 
in the conformal SYK model of \cite{Gross:2017vhb}, $ G^{\epsilon}_{\beta} $ is the \emph{bare} covariance while in our case it is the 
\emph{effective} two point function. Ergo the infrared behavior of our effective Lagrangian reproduces (a cutoffed version of) the infrared of the genuine SYK model, 
while the infrared behavior of the conformal SYK model of \cite{Gross:2017vhb} does not.

A feature of the effective Lagrangian we introduce in this paper is that it \emph{requires} the presence of the regulator $\epsilon$: in the limit $\epsilon \to 0$ the 
bare covariance diverges\footnote{This is again in contrast with the conformal SYK model of \cite{Gross:2017vhb} whose bare version does not require
a regulator $\epsilon$.}. The effective field theory fails in this limit. Below we prove that the effective Lagrangian we propose leads to a sensible theory for 
$\epsilon$ large enough. We conjecture that this is in fact the case for any $\epsilon>0$.

\section{The $\imath \epsilon$ regularization}

We consider $q$, the number of fermions, to be even and $q\ge 4$ and we denote $\Delta = 1/q$.
We posit the $\imath \epsilon$ regularization of the two point function in the SYK model:
\begin{align*}\label{eq:reg}
  G^{\epsilon}_{\beta}(\tau) & = 
  \frac{b}{2\imath \sin(\pi\Delta)}\left[  
     \frac{1}{  \left( \frac{\beta}{\pi} \sinh  \frac{\pi(\epsilon - \imath \tau)}{\beta}   \right)^{2\Delta} } 
    - \frac{1}{ \left( \frac{\beta}{\pi} \sinh \frac{\pi (\epsilon + \imath \tau)}{\beta}   \right)^{2\Delta} }  
    \right]  \crcr 
  & =   \frac{ b } { \sin(\pi\Delta)}
     \left( \frac{\beta}{\pi} \right)^{-2\Delta}   \frac{
       \sin\left( 2\Delta \arctan \frac{ \tan\frac{\pi\tau}{\beta} }{t_{\epsilon}} \right)  
    }{ \left[   \left( \sinh\frac{\pi \epsilon}{\beta}\right)^2  \left( \cos\frac{\pi\tau}{\beta}\right)^2 + \left( \cosh\frac{\pi \epsilon}{\beta}\right)^2  \left( \sin\frac{\pi\tau}{\beta}\right)^2  \right]^{\Delta} }
    \;.
\end{align*}
Observe that $ G^{\epsilon}_{\beta}(\tau) = -  G^{\epsilon}_{\beta}(-\tau)$, $  G^{\epsilon}_{\beta}(\tau) \ge 0 $ for $0\le \tau \le \beta / 2$ and that in the $\epsilon\to 0$ limit one recovers 
pointwise the conformal two point function at finite temperature  \cite{Maldacena:2016hyu}:
 \begin{align*}
  \lim_{\epsilon \to 0} G^{\epsilon}_{\beta}(\tau)  = G_{\beta}(\tau)  = b \; \frac{ \sgn(\tau) }{ \left| \frac{\beta}{\pi}  \sin  \frac{\pi \tau}{\beta}   \right|^{2\Delta} } \;.
\end{align*}  
The zero temperature version is obtained by taking $\beta \to \infty$:
\begin{align*}
  G^{\epsilon}_{\infty}(\tau) = \frac{b}{ 2\imath \sin\left( \pi \Delta\right) } \left[ \frac{1}{  ( \epsilon - \imath \tau)^{2\Delta} } - \frac{1} { ( \epsilon + \imath \tau )^{2\Delta}} \right]  \;.
\end{align*}
One can easily write down the momentum space representation at zero temperature:
\[
  G^{\epsilon}_{\infty}( \tau ) =  \frac{b}{ 2\imath \sin(\pi \Delta) \Gamma(2\Delta) }  
    \int_0^{\infty} d\omega \; \omega^{2\Delta-1} \; e^{ -\epsilon \omega } \left(  e^{ \imath \omega \tau }  -  e^{  - \imath \omega \tau } \right) \;,
\]
while the momentum space representation at finite temperature requires a bit more effort (see Appendix~\ref{app:momentum}):
\begin{align*}
 G^{\epsilon}_{\beta}(\tau) = \frac{b}{2\imath \sin(\pi \Delta)\Gamma(2\Delta)} \left(  \frac{2\pi}{\beta} \right)^{2\Delta}    
 \sum_{n >0}    \;
\frac{ \Gamma\left(  \frac{\beta}{2\pi} \omega_n +  \Delta   \right)  }{ \Gamma\left(  \frac{\beta}{2\pi} \omega_n +1 -  \Delta   \right)  }
e^{- \epsilon \omega_n}  \left[  e^{ \imath \omega_n  \tau} -  e^{ - \imath \omega_n  \tau}  \right] \;,
\end{align*}
where $\omega_n = \frac{2\pi}{\beta} \left( n+ \frac{1}{2} \right)$ denotes the fermionic Matsubara frequencies.
In particular $  G^{\epsilon}_{\beta} $ is a positive operator (as it is diagonal in momentum space and its eigenvalues are positive).
Observing that the Matsubara frequencies vary in increments of $\frac{2\pi}{\beta}$, one recovers directly the 
momentum space representation at zero temperature in the $\beta \to \infty$ limit. 

As mentioned in the introduction, it is clear that this $\imath\epsilon$ prescription is an $e^{-\epsilon |\omega|}$ frequency cutoff.
The Feynman graphs of the effective theory (each such graph represents the resummation of graphs of the bare model with arbitrary melonic insertions on the edges)
have $q$ valent vertices and effective propagators $G^{\epsilon}_{\beta}$ .
As $|\sinh(\epsilon \pm \imath \tau)|\ge \sinh(\epsilon)$, at finite temperature the amplitude of a graph with 
$E$ edges and $V$ internal vertices is bounded up to constants by:
\[
  \beta^V \frac{1}{\sinh(\epsilon)^{E}} \;.
\]
Of course this bound can be significantly improved (in particular the marginal power counting of any graph can be recovered easily).
At zero temperature the amplitudes are UV finite, but one might encounter IR divergences.

The right hand side of the Schwinger Dyson equation Eq.~\eqref{eq:sde} becomes with our regularization:
 \begin{equation}\label{eq:start}
     A^{\epsilon}_{\beta}(\tau) =  J^2 \int_{-\beta/2}^{\beta/2} du \; \; G^{\epsilon}_{\beta}(u-\tau) \left[G^{\epsilon}_{\beta}(u) \right]^{q-1} \;.
 \end{equation}
 
 Our first results is presented in the following theorem:
 \begin{theorem}\label{thm:1}
  $A^{\epsilon}_{\beta}(\tau) $ is a well defined distribution for any $\epsilon$ and:
  \begin{align}\label{eq:main}
   \lim_{\epsilon\to 0}   A_{\beta}^{\epsilon}(\tau) =  \delta(\tau) \;,
  \end{align}
  in the sense of distributions.
 \end{theorem}
 \begin{proof}
See section \ref{sec:proofsthm}  
 \end{proof}

 Observe that $A^{\epsilon}_{\beta}$ can also be viewed as a linear operator on the Hilbert space $L^2[(-\beta/2,\beta/2)]$:
 \[
  \left( A^{\epsilon}_{\beta} f\right)(\tau) = \int_{-\beta/2}^{\beta/2} d\tau' \; A^{\epsilon}_{\beta}(\tau - \tau') f(\tau') \;,
 \]
which commutes with the inverse covariance $ (G^{\epsilon}_{\beta})^{-1}$:
\begin{align*}
 \bigg(  (G^{\epsilon}_{\beta})^{-1}   A^{\epsilon}_{\beta} \bigg)(\tau_1, \tau_2) =  -  J^2 \left[G^{\epsilon}_{\beta}(\tau_1 -\tau_2) \right]^{q-1}  =   \bigg(    A^{\epsilon}_{\beta}   (G^{\epsilon}_{\beta})^{-1} \bigg)(\tau_1, \tau_2)  \;.
\end{align*}

\subsection{Effective field theory} 

One of the most interesting facts about this $\imath \epsilon$ regularization is that it allows one to introduce an effective field theory 
reproducing the IR behavior of the SYK model at all orders in $1/N$.

Our aim is to write a field theory whose effective resummed leading order two point function is the IR propagator of the SYK model  $G^{\epsilon}_{\beta}(\tau_1,\tau_2) $ 
and whose interaction that of Eq.~\ref{eq:action}. 
If we take the bare propagator of the effective field theory to be  $G^{\epsilon}_{\beta}(\tau_1,\tau_2) $,
that is if we consider the conformal SYK  model of \cite{Gross:2017vhb} with momentum cutoff, the effective two point function at leading order in $1/N$
will be $G^{\epsilon}_{\beta}(\tau_1,\tau_2) $ dressed by melonic radiative corrections. We add to the bare theory a bi local counterterm:
\[ \frac{1}{2} \int_{-\beta/2}^{\beta/2} d\tau_1 d\tau_2 \; \sum_a \chi_a(\tau_1)  {\bf A}^{\epsilon}_{\beta}(\tau_1,\tau_2)  \chi_a(\tau_2) \;,  \]
so as to precisely cancel these radiative corrections at leading order in $1/N$ and lead to an effective two point function
exactly equal to $G^{\epsilon}_{\beta}$.
In order to determine the appropriate counterterm, let us 
take for the moment some arbitrary $ {\bf A}^{\epsilon}_{\beta}$ and denote the resummed two point function in the melonic sector with this choice of counterterm $G_{{\bf A}^{\epsilon}_{\beta};\beta }^{\epsilon} $.
The SDE of this model at leading order in $1/N$ writes:
\[
 1 = G_{{\bf A}^{\epsilon}_{\beta};\beta}^{\epsilon} [ G^{\epsilon}_{\beta}]^{-1}  +  G_{{\bf A}^{\epsilon}_{\beta};\beta}^{\epsilon} {\bf A}^{\epsilon}_{\beta}
 -   G_{{\bf A}^{\epsilon}_{\beta};\beta}^{\epsilon}  {\bf \Sigma}_{{\bf A}^{\epsilon}_{\beta};\beta}^{\epsilon}   \;,\qquad  {\bf \Sigma}_{{\bf A}^{\epsilon}_{\beta};\beta}^{\epsilon}  =J^2 [ G_{{\bf A}^{\epsilon}_{\beta};\beta}^{\epsilon}]^{q} \;,
\]
where $ {\bf \Sigma}_{{\bf A}^{\epsilon}_{\beta};\beta}^{\epsilon}   $ is the self energy at melonic order in the model with counterterm.
We now require that $ G_{{\bf A}^{\epsilon}_{\beta};\beta}^{\epsilon}  = G^{\epsilon}_{\beta} $ is a solution of this equation 
which imposes:
\[
 {\bf A}^{\epsilon}_{\beta}(\tau_1,\tau_2)  = J^2 \left[ G^{\epsilon}_{\beta}(\tau_1,\tau_2)  \right]^{q-1} = - \bigg( [ G_{\beta}^{\epsilon} ]^{-1} A^{\epsilon}_{\beta} \bigg) (\tau_1,\tau_2)\;,
\]
hence the effective field theory action we propose is:
\begin{empheq}[box=\widefbox]{align}
S^{{\rm eff}} & = \frac{1}{2}\int_{-\beta/2}^{\beta/2}  d \tau_1 d\tau_2  \sum_{ a} \chi_{a}  (\tau_1)   
     \bigg(  \left[G^{\epsilon}_{\beta} \right]^{-1}  ( 1 - A^{\epsilon}_{\beta} ) \bigg)(\tau_1,\tau_2) \chi_{a} (\tau_2)  
     \crcr  & \qquad \qquad + J \sum_{a^1,\dots a^q} T_{a^1\dots a^q} \int_{-\beta/2}^{\beta/2} d\tau   \; \chi_{ a^1} (\tau) \dots   \chi_{a^q}(\tau) \;,  
\end{empheq}
and the random couplings are of course still quenched and distributed on a Gaussian.

The bare covariance of the effective SYK field theory is:
\[
  G^{\epsilon}_{\beta} \frac{1}{1 - A^{\epsilon}_{\beta}}    \; ,
\]
and is a well defined positive operator for $\epsilon$ large enough due to the following result.
\begin{theorem}\label{thm:doi}
  For any finite inverse temperature $\beta$ and for $\epsilon$ large enough such that:
   \[
    \frac{  1-t_{\epsilon}^2 }{ 1+ t_{\epsilon}^2 } \bigg[ 1  + \frac{2\Delta + 1 } {\sqrt{\pi}} t_{\epsilon} \bigg] \le  \bigg[ \tan(\pi \Delta)  \bigg]^{q-2} \;,
   \]
   the operator $ A^{\epsilon}_{\beta}$ is bounded in norm by $1$:
  \[
    || A^{\epsilon}_{\beta} ||_{\rm op} = \sup_{f, \; ||f||_{2} \le 1} || A^{\epsilon}_{\beta}f||_{2} \le 1 \;,
  \]
  where $||\cdot ||_2$ denotes the $L^2$ norm on $L^2[(-\beta/2,\beta/2)]$.
\end{theorem}
 \begin{proof}
See section \ref{sec:proofsthm2}  
 \end{proof}

The effective field theory will break down at some momentum scale. Indeed, Theorem \ref{thm:doi} ensures that the effective theory is well defined only for low enough momentum cutoff $\epsilon^{-1}$.
From Theorem \ref{thm:1} we see that the bare covariance of the model diverges in the $\epsilon \to 0$ limit, hence the effective field theory certainly breaks down in the limit. We conjecture that the effective 
field theory breaks down only in the $\epsilon \to 0$ limit, that is we conjecture that Theorem \ref{thm:doi} can be extended to any $\epsilon>0$. 

\section{The Schwinger Dyson equation}\label{sec:proofsthm}

In this section we prove Theorem~\ref{thm:1}.

Let us denote:
\begin{align*}
 & s_{\epsilon - \imath \tau} = \sinh\frac{\pi(\epsilon - \imath \tau)}{\beta} \;, \;\; c_{\epsilon - \imath \tau} = \cosh\frac{\pi(\epsilon - \imath \tau)}{\beta} \;, \;\; t_{\epsilon - \imath \tau} = \tanh\frac{\pi(\epsilon - \imath \tau)}{\beta}  \;, \crcr
 & s_{\epsilon + \imath \tau} = \sinh\frac{\pi(\epsilon + \imath \tau)}{\beta} \;, \;\; c_{\epsilon + \imath \tau} = \cosh\frac{\pi(\epsilon + \imath \tau)}{\beta} \;, \;\; t_{\epsilon + \imath \tau} = \tanh\frac{\pi(\epsilon + \imath \tau)}{\beta} \;, \crcr
 & s_{\epsilon } = \sinh\frac{\pi \epsilon  }{\beta} \;, \;\; c_{\epsilon  } = \cosh\frac{\pi \epsilon }{\beta} \;, \;\; t_{\epsilon  } = \tanh\frac{\pi \epsilon }{\beta} \;, \crcr
 & s_{\tau } = \sin\frac{\pi \tau }{\beta} \;, \;\; c_{\tau } = \cos\frac{\pi \tau }{\beta} \;, \;\; t_{\tau } = \tan\frac{\pi \tau }{\beta} \;.
\end{align*}  
 We start by rewriting $A^{\epsilon}(\tau)$ as a convergent integral more suitable to discuss the $\epsilon \to 0$ limit.

 \begin{proposition}\label{prop:formula}
  We have the following integral representation:
\begin{align}\label{eq:final}
 & A^{\epsilon}_{\beta}(\tau)   =  2\pi J^2\left( \frac{b}{2\imath \sin(\pi\Delta)} \right)^q  \left( \frac{\pi}{\beta} \right) \frac{1}{\Gamma(2\Delta)} \crcr
& \times  \bigg\{ 
  \frac{-1}{ \Gamma(2-2\Delta) }     \frac{1}{c_{\epsilon}^{2\Delta (q-1) } }
\left[ \frac{1}{   c_{\epsilon + \imath \tau}^{2\Delta}  [ t_{\epsilon + \imath \tau} + t_{\epsilon}  ] } 
    + \frac{1}{   c_{\epsilon - \imath \tau}^{2\Delta} [ t_{\epsilon - \imath \tau} + t_{\epsilon}   ] } 
    \right] + \sum_{r=1}^{q/2-1} \binom{q-1}{r} (-1)^r   \crcr
& \qquad  
\int_{1}^{\infty} \frac{ dy }{2^{2\Delta (q-1) -1} }\; \frac{ (y-1)^{2\Delta(q-1-r)-1} (y+1)^{2\Delta r-1} - (y+1)^{2\Delta(q-1-r)-1} (y-1)^{2\Delta r-1}   }{ \Gamma[ 2\Delta(q-1-r)]   \Gamma( 2\Delta r) }    
   \crcr
&\qquad  \qquad  \qquad \times   \frac{1}{c_{\epsilon}^{2\Delta (q-1) } }
\left[ \frac{1}{   c_{\epsilon + \imath \tau}^{2\Delta}  [ t_{\epsilon + \imath \tau} + t_{\epsilon} y ] } 
    + \frac{1}{   c_{\epsilon - \imath \tau}^{2\Delta} [ t_{\epsilon - \imath \tau} + t_{\epsilon} y] } 
    \right]   \bigg\} \;.
\end{align} 
  \end{proposition}
 \begin{proof}
  See Appendix~\ref{app:proofpropo}.
 \end{proof}
 
From Eq.~\eqref{eq:final}, one can show that $A^{\epsilon}(\tau)$ is a well defined distribution for any $\epsilon>0$ and that in the sense of distributions it converges to $\delta(\tau)$.
Indeed, let us consider a term in Eq.~\eqref{eq:final}. When applied on a test function $f(\tau)$ it has the generic form:
 \begin{align}\label{eq:ontest}
\left( \frac{\pi}{\beta} \right)\int_{-\beta/2}^{\beta/2} d\tau  \int_1^{\infty} dy\; H(y) \frac{1}{c_{\epsilon}^{2\Delta (q-1) } }
\left[ \frac{1}{   c_{\epsilon + \imath \tau}^{2\Delta}  [ t_{\epsilon + \imath \tau} + t_{\epsilon} y ] } 
    + \frac{1}{   c_{\epsilon - \imath \tau}^{2\Delta} [ t_{\epsilon - \imath \tau} + t_{\epsilon} y ] } 
    \right]  f(\tau) \;, 
 \end{align}
where:
\[ H(y) = \frac{1}{2^{2\Delta (q-1) -1} }\; \frac{ (y-1)^{2\Delta(q-1-r)-1} (y+1)^{2\Delta r-1} - (y+1)^{2\Delta(q-1-r)-1} (y-1)^{2\Delta r-1}   }{ \Gamma[ 2\Delta(q-1-r)]   \Gamma( 2\Delta r) }    \;, \]
 is a function such that:
\begin{itemize}
 \item $H(y)$ is integrable in $y\sim 1$,
 \item $H(y) \sim y^{2\Delta(q-1) -3 } $ for $y\sim \infty$ hence $H(y)$ is integrable for $y\sim \infty$,
 \item $H(y) \le 0$ for $y\in [ 1,\infty]$.
\end{itemize}
We now express $ c_{\epsilon \pm \imath \tau}^{-2\Delta}$ and $t_{\epsilon \pm \imath \tau}$  in terms of $t_{\tau}$ by the formulae:
\begin{align*}
&   \frac{1}{   c_{\epsilon \pm \imath \tau}^{2\Delta} }  = ( 1 + t_{\tau}^2 )^{\Delta}   
 e^{-2 \Delta \ln (  c_{\epsilon} \pm \imath s_{\epsilon} t_{\tau} ) } 
  = \frac{1}{c_{\epsilon}^{2\Delta}} \left( \frac{ 1 + t_{\tau}^2 }{ 1 + t_{\epsilon}^2 t_{\tau}^2   } \right)^{\Delta} e^{ \mp 2\imath \Delta \arctan(t_{\epsilon} t_{\tau} ) } \;, \crcr
& t_{\epsilon \pm \imath \tau} = \frac{s_{\epsilon \pm \imath \tau}}{c_{\epsilon \pm \imath \tau}} = \frac{ s_{\epsilon} \pm \imath c_{\epsilon} t_{\tau} }{    c_{\epsilon} \pm \imath s_{\epsilon} t_{\tau}  }
 = \frac{t_{\epsilon} \pm \imath t_{\tau}}{1 \pm \imath t_{\epsilon } t_{\tau}}  \;,
\end{align*}
and changing variables to  $v = \frac{t_{\tau}}{(1+y) t_{\epsilon}}$, Eq.~\eqref{eq:ontest} becomes:
\begin{align}\label{eq:almost}
 & \int_1^{\infty} dy     \int_{-\infty}^{\infty}  d v \;  R^{\epsilon}(v,y) \;  H(y) \;  f\left(\frac{\beta}{\pi} \arctan[(1+y)t_{\epsilon}v] \right) \;, \crcr
 R^{\epsilon}_{\beta}(v,y)&= \frac{ 2 (1-t_{\epsilon}^2)
    }{ \bigg[   1 + t_{\epsilon}^2 (1+y)^2 v^2 \bigg]^{1-\Delta} \bigg[  1 + t_{\epsilon}^4 (1+y)^2 v^2  \bigg]^{\Delta}
      \bigg[ 1 + v^2   (1 + y t_{\epsilon}^2)^2   \bigg]
    } \crcr
  & \qquad \times   \bigg\{  \cos\left[ 2 \Delta \arctan(t_{\epsilon}^2 (1+y) v ) \right] \bigg(  1 + t_{\epsilon}^2 v^2 (1+y)(1 + y t_{\epsilon}^2) \bigg) \crcr 
  & \qquad \qquad \qquad \qquad  + \sin \left[ 2 \Delta \arctan(t_{\epsilon}^2 (1+y) v ) \right] v (1-t_{\epsilon}^2) \bigg\}
 \;.
\end{align} 

Using Proposition \ref{prop:bound} in the Appendix~\ref{app:bound}, and denoting $||f||_{\infty}$ the $L^{\infty}$ norm of $f$ (which is a constant, if $f$ is a test function), the integral in 
Eq.~\eqref{eq:almost} is bounded by:
\begin{align*}
& 2\pi K_{\epsilon} ||f ||_{\infty} \int_{1}^{\infty} dy  \; |H(y)| \le K \;,
\end{align*}
for some constant $K$ independent of $\epsilon$ (as $K_{\epsilon}<3$). By the Lebesgue dominated convergence theorem we can then commute the $\epsilon \to 0$ limit and the integral and we have:
\begin{align*}
&  \lim_{\epsilon \to 0}
\left( \frac{\pi}{\beta} \right)\int_{-\beta/2}^{\beta/2} d\tau  \int_1^{\infty} dy\; H(y) \frac{1}{c_{\epsilon}^{2\Delta (q-1) } }
\left[ \frac{1}{   c_{\epsilon + \imath \tau}^{2\Delta}  [ t_{\epsilon + \imath \tau} + t_{\epsilon} y ] } 
    + \frac{1}{   c_{\epsilon - \imath \tau}^{2\Delta} [ t_{\epsilon - \imath \tau} + t_{\epsilon} y ] } 
    \right]  f(\tau) \crcr
 & \qquad   = 2\pi f(0) \int_1^{\infty}dy \; H(y) \;,
\end{align*}
as $\lim_{\epsilon \to 0}R^{\epsilon}_{\beta}(v,y) = \frac{2}{1+v^2} $. We therefore obtain:
\begin{align}\label{eq:monstru}
 & \lim_{\epsilon \to 0} \int_{-\beta/2}^{\beta/2} d\tau \; A^{\epsilon}_{\beta}(\tau)f(\tau)  =\crcr
&  = f(0) (2\pi J)^2\left( \frac{b}{2\imath \sin(\pi\Delta)} \right)^q   \frac{1}{\Gamma(2\Delta)}   \bigg\{ 
  \frac{-1}{ \Gamma(2-2\Delta) }   + \sum_{r=1}^{q/2-1} \binom{q-1}{r} (-1)^r   \crcr
& \qquad  
\int_{1}^{\infty} \frac{ dy }{2^{2\Delta (q-1) -1} }\; \frac{ (y-1)^{2\Delta(q-1-r)-1} (y+1)^{2\Delta r-1} - (y+1)^{2\Delta(q-1-r)-1} (y-1)^{2\Delta r-1}   }{ \Gamma[ 2\Delta(q-1-r)]   \Gamma( 2\Delta r) }  \bigg\}  \;.
\end{align}

The integrals over $y$ are evaluated in Appendix~\ref{app:auxiliary} and we get: 
\begin{align*}
& \lim_{\epsilon \to 0} \int_{-\beta/2}^{\beta/2} d\tau \; A^{\epsilon}_{\beta}(\tau)f(\tau) = f(0)
 (2\pi J)^2\left( \frac{b}{2\imath \sin(\pi\Delta)} \right)^q   \frac{1}{\Gamma(2\Delta)} \crcr
& \times \bigg\{ 
  \frac{-1}{ \Gamma(2-2\Delta) }   + \sum_{r=1}^{q/2-1} \binom{q-1}{r} (-1)^r   
   \frac{ 1 }{  \Gamma[ 2\Delta(q-1-r)]   \Gamma( 2\Delta r) }\crcr
&  \qquad \qquad  \times  \left[   
    \left( \frac{1-2\Delta r}{-1 + 2\Delta} \right) \frac{\Gamma(2\Delta)\Gamma[2-2\Delta -2\Delta r] }{\Gamma(2-2\Delta r)}
    - \left( \frac{1-2\Delta(q-1-r)}{ - 1 + 2\Delta} \right) \frac{\Gamma(2\Delta)\Gamma(2\Delta r) }{\Gamma(2\Delta + 2\Delta r)}
    \right]  \bigg\} \;.
\end{align*}
Observe that all the terms in the last two lines can be combined in a unique sum over $r$:
\begin{align*}
f(0)
 (2\pi J)^2\left( \frac{b}{2\imath \sin(\pi\Delta)} \right)^q   \frac{1}{\Gamma(2\Delta)}      \sum_{r=1}^{q-1} \binom{q-1}{r} (-1)^r   
    \left( \frac{1-2\Delta r}{-1 + 2\Delta} \right) \frac{\Gamma(2\Delta) }{ \Gamma( 2\Delta r) \Gamma(2-2\Delta r)}  \;,
\end{align*}
and using:
\begin{align*}
&  \sum_{r=1}^{q-1} \binom{q-1}{r} (-1)^r   \frac{ (1-2\Delta r)  }{ \Gamma( 2\Delta r) \Gamma(2-2\Delta r)   }  
   = \frac{1}{\pi}\sum_{r=1}^{q-1} \binom{q-1}{r} (-1)^r   \sin(2\pi \Delta r) \crcr
& \qquad =  \frac{  (2\imath)^{q} }{2\pi} \bigg[ \sin ( \pi \Delta ) \bigg]^{q-1}\cos\bigg( \pi \Delta(q-1) \bigg) \;,
\end{align*}
we finally obtain:
\begin{align*}
 \lim_{\epsilon \to 0} \int_{-\beta/2}^{\beta/2} d\tau \; A^{\epsilon}_{\beta}(\tau)f(\tau)   = f(0) \;   J^2 b^q \frac{\pi}{ \frac{1}{2} - \frac{1}{q} } \frac{ \cos\frac{\pi}{q}}{ \sin\frac{\pi}{q}} \;,
\end{align*}
which completes the proof of Theorem \ref{thm:1}.

\subsection{The bare covariance}\label{sec:proofsthm2}

We now prove Theorem~ \ref{thm:doi}.

Observe  that for any function in $L^2[(-\beta/2,\beta/2)]$, the $L^2$ norm is bounded by the $L^{\infty}$ norm $||f||_2 \le \beta^{1/2} ||f||_{\infty}$, hence: 
\[
  ||A^{\epsilon}_{\beta}||_{\rm op} =  \sup_{f, \;  ||f||_2 \le 1} ||A^{\epsilon}_{\beta} f||_2 \le \sup_{f, \;  ||f||_{\infty} \le \beta^{-1/2}} ||A^{\epsilon}_{\beta} f||_2 \; . 
\]
 On the other hand:
\begin{align*}
 \bigg( || A^{\epsilon}_{\beta} f ||_2 \bigg)^2 &=  \int_{-\beta/2}^{\beta/2} 
 d\tau_1  \int_{-\beta/2}^{\beta/2} d\tau   \;  \bar A^{\epsilon}_{\beta}(\tau ) \bar f(\tau_1 + \tau)  \int_{-\beta/2}^{\beta/2} d\tau'   \;     A^{\epsilon}_{\beta}(\tau' ) f(\tau_1 + \tau')   \crcr 
 & \le \beta ||f||_{\infty}^2  \left( \int_{-\beta/2}^{\beta/2} d\tau   \; | A^{\epsilon}_{\beta}(\tau )| \right)^2 
  \Rightarrow  ||A^{\epsilon}||_{op} \le   \int_{-\beta/2}^{\beta/2} d\tau   \; \left| A^{\epsilon}_{\beta}(\tau ) \right| \;,
\end{align*}
therefore, in the notation of Section~\ref{sec:proofsthm}, the operator norm of $ A^{\epsilon}_{\beta} $ is bounded by a sum of terms of the form:
\[
 \int_1^{\infty} dy     \int_{-\infty}^{\infty}  d v \;  R^{\epsilon}(v,y) \; | H(y)| \;,
\]
therefore we obtain a bound:
\begin{align*}
  ||A^{\epsilon}_{\beta}||_{\rm op} & \le K_{\epsilon} 
 (2\pi J)^2\left( \frac{b}{2 \sin(\pi\Delta)} \right)^q   \frac{1}{\Gamma(2\Delta)}    \sum_{r=1}^{q-1} \binom{q-1}{r}   
   \bigg| \left( \frac{1-2\Delta r}{ 1 - 2\Delta} \right) \frac{\Gamma(2\Delta) }{ \Gamma( 2\Delta r) \Gamma(2-2\Delta r)}
     \bigg|  \crcr
     &  = K_{\epsilon} \frac{  1 }{  \cos(\pi  \Delta) \bigg[2 \sin(\pi  \Delta)\bigg]^{q-1} } \sum_{r=1}^{q-1} \binom{q-1}{r}      \sin(2\pi \Delta r) 
     =\frac{ K_{\epsilon}  }{   \bigg[ \tan(\pi  \Delta)\bigg]^{q-2} } \;.
\end{align*}

\section*{Acknowledgements}

The author would like to thank Igor Klebanov and Grigory Tarnopolsky for the numerous discussions 
on divergences in the SYK model which inspired this project. The author also thanks Guillaume Bossard
for discussions on the effective field theory part of this paper.

This research was supported in part by Perimeter Institute for Theoretical Physics. Research at
Perimeter Institute is supported by the Government of Canada through the Department of Innovation,
Science and Economic Development Canada and by the Province of Ontario through the Ministry of
Research, Innovation and Science.

\newpage

\appendix

\section{The momentum space representation}\label{app:momentum}

At finite temperature we use the Fourier transform conventions:
\[
 \tilde f(\omega) = \int_{-\beta/2}^{\beta/2} d\tau \; e^{\imath \omega \tau}   f(\tau) \;,\qquad f(\tau)  = \frac{1}{\beta} \sum_{n\in \mathbb{Z}} e^{-\imath \omega_n \tau } \tilde f(\omega_n) \;,
\]
where $\omega_n = \frac{2\pi}{\beta}\left( n+\frac{1}{2} \right)$ are the fermionic Matsubara frequencies.
 Our aim in this section is to compute the Fourier transform:
\begin{align*}
 \tilde G^{\epsilon}_{\beta}(\omega) = \frac{b}{2\imath \sin(\pi \Delta)} \left( \frac{\pi}{\beta} \right)^{2\Delta}\int_{-\beta/2}^{\beta/2} d\tau \; e^{\imath \omega \tau}  
  \left[ 
 \frac{1}{ \left(  \sinh  \frac{\pi (\epsilon - \imath \tau)  }{\beta}    \right)^{2\Delta} }  -
  \frac{1}{ \left(  \sinh  \frac{\pi (\epsilon + \imath \tau) }{\beta}    \right)^{2\Delta} } 
 \right] \;,
\end{align*}
where $\omega$ is one of the fermionic Matsubara frequencies $\omega_n$.

We denote $s_{\epsilon} = \sinh \frac{\pi \epsilon}{\beta}$, $c_{\epsilon} = \cosh\frac{\pi \epsilon}{\beta}$, $t_{\epsilon} = \tanh\frac{\pi \epsilon}{\beta}$
and change variables to $t = \tan\frac{\pi \tau}{\beta} $ to get:
 \begin{align}\label{eq:full}
 &  \frac{b}{2\imath \sin(\pi \Delta)}  \left(  \frac{\pi}{\beta} \right)^{2\Delta-1} \int_{-\infty}^{\infty} \frac{dt}{ 1+t^2} \left( \frac{1+\imath t}{1-\imath t}\right)^{ \frac{\beta}{2\pi}\omega} (1+t^2)^{\Delta}
 \left[ \frac{ 1} { \left(  s_{\epsilon}  - \imath   c_{\epsilon} t   \right)^{2\Delta}  } 
   -   \frac{1} {
  \left( s_{\epsilon}   + \imath   c_{\epsilon} t \right)^{2\Delta} } 
  \right] \\
 & \qquad =  \frac{b}{2\imath \sin(\pi \Delta)}  \left(  \frac{\pi}{\beta} \right)^{2\Delta-1} (-\imath )\int_{-\imath\infty}^{\imath\infty}   dz \; \; \frac{( 1+ z)^{ \frac{\beta}{2\pi}\omega - 1 +\Delta } }{ (1-z)^{ \frac{\beta}{2\pi}\omega + 1 - \Delta  } }
 \left[ 
 \left(  \frac{1} {
  s_{\epsilon}    -   c_{\epsilon}  z } \right)^{2\Delta} 
   - \left(  \frac{1} {
  s_{\epsilon}    +   c_{\epsilon}  z} \right)^{2\Delta}
  \right] \;. \nonumber
 \end{align}
 At $z\sim \infty$ the integrand behaves like $z^{-2}$ hence we can turn the contour of integration on $z$ to run around the positive or negative real axis. 
  
 Let us consider $\omega >0$ (the case $\omega<0$ is similar). The first term in Eq.~\eqref{eq:full} writes:
  \begin{align*}
   \frac{b}{2\imath \sin(\pi \Delta)} \left(  \frac{\pi}{\beta} \right)^{2\Delta-1} (-\imath )\int_{-\imath\infty}^{\imath\infty}   dz \; \; \frac{( 1+ z)^{ \frac{\beta}{2\pi}\omega - 1 +\Delta } }{ (1-z)^{ \frac{\beta}{2\pi}\omega + 1 - \Delta  } }
 \left(  \frac{1} {
  s_{\epsilon}    -   c_{\epsilon}  z } \right)^{2\Delta}  \;,
  \end{align*}
having singularities at $z =\pm 1,t_{\epsilon}$. We turn the contour to run along the negative real axis.
The only factor which has a discontinuity is $(1+z)^{ \frac{\beta}{2\pi}\omega - 1 +\Delta }$ and we obtain:
\begin{align*}
 &\lim_{\delta \to 0} \int_{ - 1}^{ - \infty} dy  \; \;  
 \frac{1 }{ (1 -y )^{ \frac{\beta}{2\pi}\omega + 1 - \Delta  } }
 \left(  \frac{1} { s_{\epsilon}    -   c_{\epsilon}  y  } \right)^{2\Delta}  
  \bigg\{ ( 1 + y + \imath \delta )^{ \frac{\beta}{2\pi}\omega - 1 +\Delta }   
  - ( 1 +  y - \imath \delta )^{ \frac{\beta}{2\pi}\omega - 1 +\Delta }  
  \bigg\} \crcr
 & =  - \int_{1}^{\infty} dx  \; \;  
 \frac{1 }{ (1 + x )^{ \frac{\beta}{2\pi}\omega + 1 - \Delta  } }
 \left(  \frac{1} { s_{\epsilon}    +   c_{\epsilon}  x  } \right)^{2\Delta}
     (x-1)^{  \frac{\beta}{2\pi}\omega - 1 +\Delta }
     \left[ e^{ \left( \frac{\beta}{2\pi}\omega - 1 +\Delta \right) (\imath \pi) } -e^{ \left( \frac{\beta}{2\pi}\omega - 1 +\Delta \right)(-\imath \pi) } \right] \;.
\end{align*}
Recalling that $\frac{\beta}{2\pi}\omega = n+1/2$, we have $e^{ \imath (n-1/2  + \Delta ) \pi } -  e^{  - \imath (n-1/2  + \Delta ) \pi } = 2\imath (-1)^{n+1} \cos(\Delta \pi)$
and finally the first term in Eq.~\eqref{eq:full} becomes:
\begin{align}\label{eq:t1}
  \frac{b}{2\imath \sin(\pi \Delta)}  \left(  \frac{\pi}{\beta} \right)^{2\Delta-1}   2 (-1)^{n} \cos(\Delta \pi) \int_{1}^{\infty} dx  \; \;  
 \frac{1 }{ (1 + x )^{ \frac{\beta}{2\pi}\omega + 1 - \Delta  } }
 \left(  \frac{1} { s_{\epsilon}    +   c_{\epsilon}  x  } \right)^{2\Delta}
     (x-1)^{  \frac{\beta}{2\pi}\omega - 1 +\Delta } \;.
\end{align}
Observe that this integral is convergent both for $x\sim \infty$ and for $x\sim 1$.
We now consider the second term in Eq.~\eqref{eq:full}:
  \begin{align*}
  \frac{b}{2\imath \sin(\pi \Delta)}   \left(  \frac{\pi}{\beta} \right)^{2\Delta-1}  \imath  \int_{-\imath\infty}^{\imath\infty}   dz \; \; \frac{( 1+ z)^{ \frac{\beta}{2\pi}\omega - 1 +\Delta } }{ (1-z)^{ \frac{\beta}{2\pi}\omega + 1 - \Delta  } }
 \left(  \frac{1} { s_{\epsilon}    +  c_{\epsilon}  z } \right)^{2\Delta}  \;,
  \end{align*}
having singularities at $z = \pm 1, - t_{\epsilon}$. We close again the contour around the negative real axis to obtain:
\begin{align*}
&  \lim_{\delta \to 0} \int_{ - t_{\epsilon}}^{- \infty} dy  \; \;    \frac{1}{ ( 1 -y )^{ \frac{\beta}{2\pi}\omega + 1  - \Delta } } 
 \bigg\{   
   \frac{ ( 1+ y+ \imath \delta )^{ \frac{\beta}{2\pi}\omega - 1 +\Delta } } { \left( s_{\epsilon}    +   c_{\epsilon}  y  + \imath c_{\epsilon}\delta  \right)^{2\Delta}   } 
  -  \frac{ ( 1+ y - \imath \delta )^{ \frac{\beta}{2\pi}\omega - 1 +\Delta } } {
 \left(  s_{\epsilon}    +   c_{\epsilon}  y - \imath c_{\epsilon}\delta \right)^{2\Delta}   } 
  \bigg\} \crcr
& =  -  \lim_{\delta \to 0} \int_{   t_{\epsilon}}^{  \infty} dx  \; \;  
  \frac{1}{ ( 1 + x )^{ \frac{\beta}{2\pi}\omega + 1  - \Delta } } 
 \bigg\{ 
  \frac{ ( 1 - x+ \imath \delta )^{ \frac{\beta}{2\pi}\omega - 1 +\Delta }  } {\left(   s_{\epsilon}    -   c_{\epsilon}  x  + \imath c_{\epsilon}\delta \right)^{2\Delta}   } 
  - \frac{ ( 1 - x - \imath \delta )^{ \frac{\beta}{2\pi}\omega - 1 +\Delta }} {
   \left(  s_{\epsilon}    -   c_{\epsilon}  x - \imath c_{\epsilon}\delta  \right)^{2\Delta}  }
  \bigg\} \;.
\end{align*}
The integral splits into an integral over the interval $(t_{\epsilon},1) $ and a second integral over the interval $(1,\infty)$ (as $t_{\epsilon}<1$). Taking the limit $\delta \to 0$ the first integral contributes:
\begin{align*}
 \int_{t_{\epsilon}}^1 dx  \; \;  
  \frac{1}{ ( 1 + x )^{ \frac{\beta}{2\pi}\omega + 1  - \Delta } }( 1 - x )^{ \frac{\beta}{2\pi}\omega - 1 +\Delta } \frac{1}{ ( c_{\epsilon} x -s_{\epsilon} )^{2\Delta}  }
   \left[ e^{-  2\Delta  (\imath \pi)} -  e^{-  2\Delta  (-\imath \pi)} \right] \;,
\end{align*}
while the second one is:
\begin{align*}
 \int_{1}^{\infty} dx  \; \;  
  \frac{1}{ ( 1 + x )^{ \frac{\beta}{2\pi}\omega + 1  - \Delta } }\frac{( x - 1 )^{ \frac{\beta}{2\pi}\omega - 1 +\Delta } }{ ( c_{\epsilon} x -s_{\epsilon} )^{2\Delta}  }
   \left[ e^{ \left(\frac{\beta}{2\pi}\omega - 1 +\Delta  \right) (\imath \pi) }  e^{-  2\Delta  (\imath \pi)} - 
    e^{ \left(\frac{\beta}{2\pi}\omega - 1 +\Delta  \right)(-\imath \pi) }  
   e^{-  2\Delta  (-\imath \pi)} \right] \;.
\end{align*}
Recalling again that $\omega = \frac{2\pi}{\beta} (n+1/2) $, we have
$ e^{\imath (n-1/2  - \Delta) \pi} - e^{ - \imath (n-1/2 - \Delta) \pi} = 2  \imath(-1)^{n+1} \cos(\pi \Delta)
$, hence finally the second term in Eq.~\eqref{eq:full} is:
\begin{align}\label{eq:t2}
&   \frac{b}{2\imath \sin(\pi \Delta)}    \left(  \frac{\pi}{\beta} \right)^{2\Delta-1} (-  2 )  \sin(2\pi \Delta)  \int_{t_{\epsilon}}^1 dx  \; \;  
  \frac{1}{ ( 1 + x )^{ \frac{\beta}{2\pi}\omega + 1  - \Delta } } \frac{( 1 - x )^{ \frac{\beta}{2\pi}\omega - 1 +\Delta }}{ ( c_{\epsilon} x -s_{\epsilon} )^{2\Delta}  } \crcr
& \qquad + \frac{b}{2\imath \sin(\pi \Delta)}  \left(  \frac{\pi}{\beta} \right)^{2\Delta-1}   2  (-1)^{n+1} \cos(\pi \Delta)
\int_{1}^{\infty} dx  \; \;  
  \frac{1}{ ( 1 + x )^{ \frac{\beta}{2\pi}\omega + 1  - \Delta } }\frac{( x - 1 )^{ \frac{\beta}{2\pi}\omega - 1 +\Delta } }{ ( c_{\epsilon} x -s_{\epsilon} )^{2\Delta}  } \;.
\end{align}

Adding up Eq.~\eqref{eq:t1} and  Eq.~\eqref{eq:t2} the integrals from $1$ to $\infty$ cancel and we obtain:
\begin{align*}
 \tilde G^{\epsilon}_{\beta}(\omega) = \frac{b}{2\imath \sin(\pi \Delta)} \left(  \frac{\pi}{\beta} \right)^{2\Delta-1}  (-2) \sin(2\pi \Delta)  \int_{t_{\epsilon}}^1 dx  \; \;  
  \frac{1}{ ( 1 + x )^{ \frac{\beta}{2\pi}\omega + 1  - \Delta } } \frac{( 1 - x )^{ \frac{\beta}{2\pi}\omega - 1 +\Delta }}{ ( c_{\epsilon} x -s_{\epsilon} )^{2\Delta}  } \;,
\end{align*}
which is an absolutely convergent integral. We now change variable to $x = \tanh(s + \frac{\pi}{\beta} \epsilon)$ and obtain:
\begin{align*}
 \tilde G^{\epsilon}_{\beta}(\omega)  
  =    \frac{b}{2\imath \sin(\pi \Delta)} \left(  \frac{\pi}{\beta} \right)^{2\Delta-1}  (-2) \sin(2\pi \Delta)  \; e^{-\omega \epsilon}
     \int_{0}^{\infty} ds \; e^{-\frac{\beta}{\pi} \omega s} \frac{1}{ \left[ \sinh \left( s  \right)  \right]^{2\Delta}} \;,
\end{align*}
and changing again variables to $y = e^{-2s}$, the integral can be explicitly evaluated in terms of an Euler beta function with positive arguments:
\begin{align*}
 \tilde G^{\epsilon}_{\beta}(\omega)  
  =  \frac{b}{2\imath \sin(\pi \Delta)} \left(  \frac{2\pi}{\beta} \right)^{2\Delta-1}  (-2) \sin(2\pi \Delta)  
\frac{ \Gamma\left(  \frac{\beta}{2\pi} \omega +  \Delta   \right) \Gamma(1-2\Delta)  }{ \Gamma\left(  \frac{\beta}{2\pi} \omega +1 -  \Delta   \right)  }  e^{-\omega \epsilon}   \;.
\end{align*}
 
\section{Proof of the Proposition~\ref{prop:formula} }\label{app:proofpropo}

Substituting $ G^{\epsilon}_{\beta} $ in Eq.~\eqref{eq:start} we get:
\begin{align*}
 & A^{\epsilon}(\tau)  
 = J^2\left( \frac{b}{2\imath \sin(\pi\Delta)} \right)^q  \left( \frac{\pi}{\beta} \right)^{2\Delta q} \int_{-\beta/2}^{\beta/2} du  \crcr
 & 
   \times \left[  
     \frac{1}{  \left(   \sinh  \frac{\pi[\epsilon - \imath (u-\tau) ]}{\beta}   \right)^{2\Delta} } 
    - \frac{1}{ \left(   \sinh \frac{\pi [\epsilon + \imath (u-\tau) ]}{\beta}   \right)^{2\Delta} }  
    \right]      \left[  
     \frac{1}{  \left(   \sinh  \frac{\pi(\epsilon - \imath u)}{\beta}   \right)^{2\Delta} } 
    - \frac{1}{ \left(  \sinh \frac{\pi (\epsilon + \imath u)}{\beta}   \right)^{2\Delta} }  
    \right]^{q-1}  \;.
\end{align*}
Recalling that $\sinh(z \pm \imath x)  = \sinh(z) \cos(x) \pm \imath \cosh(z) \sin(x)$, $\Delta q =1$, changing variable to $ t = \tan\frac{\pi u}{\beta}$ and expanding the 
binomial, $A^{\epsilon}(\tau)$ becomes:
\begin{align*}
 & J^2\left( \frac{b}{2\imath \sin(\pi\Delta)} \right)^q  \left( \frac{\pi}{\beta} \right)  \sum_{r=0}^{q/2-1} \binom{q-1}{r} (-1)^r 
    \int_{-\infty}^{\infty} dt \; \;
   \left[  
     \frac{1}{  \left[  s_{\epsilon + \imath\tau}  - \imath c_{\epsilon + \imath \tau} t  \right]^{2\Delta} } 
    - \frac{1}{ \left[  s_{\epsilon - \imath\tau}  + \imath c_{\epsilon - \imath \tau} t  \right]^{2\Delta} }  
    \right]  
 \crcr
 & \qquad \times
    \left[  
     \frac{1}{ \left( s_{\epsilon} - \imath c_{\epsilon} t   \right)^{2\Delta(q-1-r)} } 
    \frac{1}{ \left( s_{\epsilon} + \imath c_{\epsilon} t   \right)^{2\Delta r} }   -      \frac{1}{ \left( s_{\epsilon} - \imath c_{\epsilon} t   \right)^{2\Delta r} } 
    \frac{1}{ \left( s_{\epsilon} + \imath c_{\epsilon} t   \right)^{2\Delta (q-1-r)} }  
    \right]  \;.
\end{align*}
Taking into account that:
 \[
  \Re(t_{\epsilon \pm \imath \tau}) = \frac{ s_{\epsilon} c_{\epsilon}}{   \cos(\frac{\pi\tau}{\beta})^2 + s_{\epsilon}^2  } >0 \;, \qquad t_{\epsilon } >0 \;,
 \]
one can use (absolutely convergent) Schwinger parametric representations to rewrite $A^{\epsilon}(\tau)$ as:
\begin{align*}
 & J^2\left( \frac{b}{2\imath \sin(\pi\Delta)} \right)^q  \left( \frac{\pi}{\beta} \right)  \sum_{r=0}^{q/2-1} \binom{q-1}{r} (-1)^r 
    \int_{-\infty}^{\infty} dt \int_{0}^{\infty} d\alpha \, d\alpha_1 d\alpha_2 \;\frac{\alpha^{2\Delta -1} \alpha_1^{2\Delta(q-1-r)-1}\alpha_2^{2\Delta r-1} }{ \Gamma( 2\Delta)   \Gamma[ 2\Delta(q-1-r)]   \Gamma( 2\Delta r) }  \crcr
&  \qquad \times  \left[  \frac{1}{  c_{\epsilon + \imath \tau}^{2\Delta}} e^{ - \alpha t_{\epsilon + \imath \tau}  + \imath \alpha t }
    -  \frac{1}{  c_{\epsilon - \imath \tau}^{2\Delta}} e^{ - \alpha t_{\epsilon - \imath \tau}  - \imath \alpha t }
    \right]   \frac{1}{c_{\epsilon}^{2\Delta (q-1) } } e^{-t_{\epsilon} (\alpha_1  + \alpha_2 ) } \bigg(  e^{\imath  t (\alpha_1 -\alpha_2 ) } - e^{ - \imath  t (\alpha_1 -\alpha_2 ) }  \bigg) \;,
\end{align*}
where the integral over $\alpha_2 $ is absent for $r=0$. The integral over $t$ can now be computed and we get:
\begin{align*}
 & 2\pi J^2\left( \frac{b}{2\imath \sin(\pi\Delta)} \right)^q  \left( \frac{\pi}{\beta} \right)  \sum_{r=0}^{q/2-1} \binom{q-1}{r} (-1)^r 
    \int_{0}^{\infty} d\alpha \, d\alpha_1 d\alpha_2 \;\frac{\alpha^{2\Delta -1} \alpha_1^{2\Delta(q-1-r)-1}\alpha_2^{2\Delta r-1} }{ \Gamma( 2\Delta)   \Gamma[ 2\Delta(q-1-r)]   \Gamma( 2\Delta r) }  \crcr
&  \qquad \times  
\frac{1}{c_{\epsilon}^{2\Delta (q-1) } } e^{-t_{\epsilon} (\alpha_1  + \alpha_2 ) }
\left[  \frac{1}{  c_{\epsilon + \imath \tau}^{2\Delta}} e^{ - \alpha t_{\epsilon + \imath \tau}   }
    +  \frac{1}{  c_{\epsilon - \imath \tau}^{2\Delta}} e^{ - \alpha t_{\epsilon - \imath \tau}   }
    \right]    \bigg[ \delta(\alpha +\alpha_1 -\alpha_2) -\delta(\alpha -\alpha_1 + \alpha_2) \bigg] \;.
\end{align*}
 Changing variables to $\alpha_1 = \alpha U$ and $\alpha_2 = \alpha V$ and integrating over $\alpha$ we get:
 \begin{align*}
   & 2\pi J^2\left( \frac{b}{2\imath \sin(\pi\Delta)} \right)^q  \left( \frac{\pi}{\beta} \right)  \sum_{r=0}^{q/2-1} \binom{q-1}{r} (-1)^r 
   \frac{1}{\Gamma(2\Delta)} \int_{0}^{\infty}  dU dV \;\frac{U^{2\Delta(q-1-r)-1} V^{2\Delta r-1} }{ \Gamma[ 2\Delta(q-1-r)]   \Gamma( 2\Delta r) }    
   \crcr
& \qquad \times   \frac{1}{c_{\epsilon}^{2\Delta (q-1) } }
\left[ \frac{1}{   c_{\epsilon + \imath \tau}^{2\Delta}  [ t_{\epsilon + \imath \tau} + t_{\epsilon} (U+V) ] } 
    + \frac{1}{   c_{\epsilon - \imath \tau}^{2\Delta} [ t_{\epsilon - \imath \tau} + t_{\epsilon} (U+V) ] } 
    \right]  \bigg[ \delta(1 +U -V) -\delta(1 -U + V) \bigg] \;,
 \end{align*}
where we recall that for $r=0$ the integral over $V$ is absent. Integrating once using the $\delta$ functions we obtain:
\begin{align*}
&  2\pi J^2\left( \frac{b}{2\imath \sin(\pi\Delta)} \right)^q  \left( \frac{\pi}{\beta} \right) \frac{1}{\Gamma(2\Delta)} 
 \bigg\{ 
  \frac{-1}{ \Gamma(2-2\Delta) }     \frac{1}{c_{\epsilon}^{2\Delta (q-1) } }
\left[ \frac{1}{   c_{\epsilon + \imath \tau}^{2\Delta}  [ t_{\epsilon + \imath \tau} + t_{\epsilon}  ] } 
    + \frac{1}{   c_{\epsilon - \imath \tau}^{2\Delta} [ t_{\epsilon - \imath \tau} + t_{\epsilon}   ] } 
    \right] \crcr
& \qquad + \sum_{r=1}^{q/2-1} \binom{q-1}{r} (-1)^r   
\int_{1}^{\infty} dV \; \frac{ (V-1)^{2\Delta(q-1-r)-1} V^{2\Delta r-1} - V^{2\Delta(q-1-r)-1} (V-1)^{2\Delta r-1} }{ \Gamma[ 2\Delta(q-1-r)]   \Gamma( 2\Delta r) }    
   \crcr
&\qquad  \qquad  \qquad \times   \frac{1}{c_{\epsilon}^{2\Delta (q-1) } }
\left[ \frac{1}{   c_{\epsilon + \imath \tau}^{2\Delta}  [ t_{\epsilon + \imath \tau} + t_{\epsilon} (2V-1) ] } 
    + \frac{1}{   c_{\epsilon - \imath \tau}^{2\Delta} [ t_{\epsilon - \imath \tau} + t_{\epsilon} (2V-1) ] } 
    \right]   \bigg\} \;,
\end{align*}
and finally, changing variables to $y = 2 V-1$ proves Proposition~\ref{prop:formula}

\section{Bound on the integral in Eq.~\eqref{eq:almost}}\label{app:bound}
 
\begin{proposition}\label{prop:bound}
For any $y\ge 1$ and $\Delta \le \frac{1}{2}$ we have:
\begin{align*}
 \int_{-\infty}^{\infty} dv\; R^{\epsilon}(v,y)  \le 2\pi K_{\epsilon} \; ,\qquad K_{\epsilon} =  \frac{  1-t_{\epsilon}^2 }{ 1+ t_{\epsilon}^2 } \bigg[ 1  + \frac{2\Delta + 1 } {\sqrt{\pi}} t_{\epsilon} \bigg] \;.
\end{align*}
\end{proposition}

\begin{proof}  Let us first find a bound for the integral:
\begin{align*}
I_{\Delta} = \int_{-\infty}^{\infty} \frac{dz}{ (z^2 + 1)^{1-\Delta} }   \;.
\end{align*}
Observe that $\lim_{\Delta \to 0} I_{\Delta}= \pi $ and the integral is convergent for $0\le \Delta \le \frac{1}{2}$.
The integrand has two cuts,$(\imath  , \imath \infty)$ and $ ( - \imath  , - \imath \infty) $. We deform the contour of integration to run around the cut $(\imath  , \imath \infty)$ and the integral becomes:
\begin{align*}
& \int_{1}^{\infty} (\imath d\rho) \; \lim_{\epsilon \to 0} \left[ e^{- (1-\Delta) \ln [ 1+  (\imath \rho + \epsilon)^2  ] } -   e^{- (1-\Delta) \ln [ 1 +  (\imath \rho - \epsilon)^2  ] } \right] \crcr
& = 2\sin \bigg[ (1-\Delta) \pi \bigg]  \int_{1}^{\infty} \frac{ d\rho }{ (\rho^2-1)^{ 1-\Delta }  } =_{\rho = y^{-1/2}} \crcr
&  = 
\sin \bigg[ (1-\Delta) \pi \bigg] \int_0^1  dy  \; y^{-\frac{1}{2} -\Delta}  (1-y)^{-( 1-\Delta ) } =
\sin \bigg[ (1-\Delta) \pi \bigg] \frac{  \Gamma \left(  \frac{1}{2}-\Delta \right) \Gamma(\Delta) }{ \Gamma\left(\frac{1}{2} \right) }
\end{align*}
Since $\Gamma\left(\frac{1}{2} \right) = \sqrt{\pi} $ and $\Gamma(x) \Gamma(1-x) = \frac{\pi}{\sin(\pi x)}$, we have:
\begin{align*}
I_{\Delta} \le  \sqrt{\pi} \frac{ \Gamma\left( \frac{1}{2} -\Delta\right) }{\Gamma(1-\Delta) } \le \sqrt{\pi}
\end{align*}
as $\Delta \le \frac{1}{2}$ and $\Gamma$ is strictly increasing for positive real arguments.

Now, going back to $ R^{\epsilon}_{\beta}(v,y)$, we use:
\begin{align*}
   \cos\left[ 2 \Delta \arctan(t_{\epsilon}^2 (1+y) v ) \right] \le 1 \;, \qquad
    \sin \left[ 2 \Delta \arctan(t_{\epsilon}^2 (1+y) v ) \right] \le 2 \Delta t_{\epsilon}^2 (1+y) v \;, 
\end{align*}
to obtain a bound (observe that $R^{\epsilon}_{\beta}(v,y) \ge 0 $):
\begin{align*}
  R^{\epsilon}(v,y)   & \le \frac{ 2 (1-t_{\epsilon}^2)
 \bigg\{     1 + v^2 t_{\epsilon}^2(1+y) (1+yt_{\epsilon}^2)     +     2  \Delta  t_{\epsilon}^2(1+y) v^2   (1-t_{\epsilon}^2)         \bigg\}
    }{ \bigg[   1 + t_{\epsilon}^2 (1+y)^2 v^2 \bigg]^{1-\Delta} \bigg[  1 + t_{\epsilon}^4 (1+y)^2 v^2  \bigg]^{\Delta}
      \bigg[ 1 + v^2   (1 + y t_{\epsilon}^2)^2   \bigg]
    } \crcr
  & \le \frac{ 2(1-t_{\epsilon}^2) }{ 1+yt_{\epsilon}^2 }
 \bigg[ \frac{1+yt_{\epsilon}^2}{1 + v^2   (1 + y t_{\epsilon}^2)^2  } +  (2\Delta+1) t_{\epsilon }  \frac{ t_{\epsilon} (1+y)   }{ \big[   1 + t_{\epsilon}^2 (1+y)^2 v^2 \big]^{1-\Delta} } \bigg] \;,
\end{align*}
therefore:
\[
 \int_{-\infty}^{\infty} dv\; R^{\epsilon}(v,y)  \le \frac{ 2(1-t_{\epsilon}^2) }{ 1+yt_{\epsilon}^2 } \bigg[ \pi + t_{\epsilon} (2\Delta +1) \sqrt{\pi}\bigg] \le
  2\pi \frac{  1-t_{\epsilon}^2 }{ 1+ t_{\epsilon}^2 } \bigg[ 1  + \frac{2\Delta + 1 } {\sqrt{\pi}} t_{\epsilon} \bigg] \;.
\]
\end{proof}

\section{The integrals in Eq.~\eqref{eq:monstru}}\label{app:auxiliary}

\begin{proposition}\label{prop:int}
For $\Re(a)>0, \Re(b)>0$ and $\Re(a+b)<2$, $a+b\neq 1$ we have:
\begin{align*}
& F(a,b) =  \int_{1}^{\infty} \frac{ dy }{2^{a+b -1} }\; \bigg[   (y-1)^{a-1} (y+1)^{b-1} - (y+1)^{a-1} (y-1)^{b-1}    \bigg]  \crcr
& =  \left( \frac{1-b}{1-a-b} \right) \frac{\Gamma(2-a-b)\Gamma(a) }{\Gamma(2-b)} - \left( \frac{1-a}{1-a-b} \right) \frac{\Gamma(2-a-b)\Gamma(b) }{\Gamma(2-a)} \;.
\end{align*}
\end{proposition}
\begin{proof}
The integral is clearly convergent in $0$. At infinity, due to the subtraction, the integrand behaves like $ y^{a+b-3}$, hence the integral converges for $\Re(a+b)<2$.
Changing variables to $x = \frac{2}{1+y} $, the integral becomes
 \begin{align*}
& \frac{1}{2^{a+b-1}} \int_1^0   \left( - 2 \frac{ dx}{x^2}\right) \left[ \left( \frac{2}{x} -2\right)^{a-1} \left( \frac{2}{x} \right)^{b-1} - \left( \frac{2}{x}\right)^{a-1} \left( \frac{2}{x} -2\right)^{b-1}  \right] \crcr
& = \int_0^1 dx   \; x^{-a-b} \left[ (1-x)^{a-1}  - (1-x)^{b-1} \right] \;.
 \end{align*}
Observe that the two terms can not be  integrated separately, as each integral would diverge in $x\sim 0$. However, the difference is convergent in $x\sim 0$ as the behavior is tamed by the explicit subtraction.
We observe that 
\[ x^{-a-b}  = \frac{1}{1-a-b} [x^{1-a-b} ]' (1-x) + x^{1-a-b} \;, \]
hence we get:
\begin{align*}
 F(a,b) =&  \left[ \frac{x^{1-a-b}  }{1-a-b}   \left[ (1-x)^{a }  - (1-x)^{b } \right] \right]_{0}^1  + \frac{1}{1-a-b} \int_0^1 dx   \; x^{1-a-b} \left[ a (1-x)^{a-1}  - b (1-x)^{b-1} \right]  \crcr
    & + \int_0^1 dx   \; x^{1 -a-b} \left[ (1-x)^{a-1}  - (1-x)^{b-1} \right] \;.
\end{align*}
As $\Re(a+b)<2$ the boundary terms cancel and all the integrals are convergent and can be expressed in terms of Euler $\Gamma$ functions.
\end{proof}

 \bibliography{/home/razvan/Desktop/lucru/Ongoing/Refs/Refs.bib}
 
\end{document}